\documentclass[preprint,12pt]{elsarticle}
\usepackage{etoolbox}
\makeatletter
\def\ps@pprintTitle{%
	\let\@oddhead\@empty
	\let\@evenhead\@empty
	\def\@oddfoot{\reset@font\hfil\thepage\hfil}
	\let\@evenfoot\@oddfoot
}
\makeatother
\usepackage{graphics}
\usepackage[table,xcdraw]{xcolor}
\usepackage{amsmath, amsthm, amssymb}
\usepackage{natbib}
\usepackage[colorlinks=true]{hyperref}
\usepackage{lscape}
\usepackage[utf8]{inputenc} 
\usepackage{amsthm}

\usepackage{enumerate}
\usepackage{mathrsfs}
\usepackage{setspace}
\usepackage{multirow}
\usepackage{graphicx}
\usepackage{amsfonts}
\usepackage{verbatim}
\usepackage{amssymb}
\usepackage{amsthm}
\usepackage{diagbox}
\usepackage{booktabs}
\usepackage{array}
\usepackage{multirow,bigstrut,bigdelim}
\newtheorem{thm}{Theorem}[section]

\theoremstyle{plain}
\newtheorem{defi}{Definition}[section]
\newtheorem{ex}{Example}[section]

\theoremstyle{remark}

\numberwithin{equation}{section}
\usepackage[mathscr]{euscript}
\usepackage{geometry} 
\usepackage{graphicx,lscape} 
\usepackage{booktabs}
\usepackage{array}
\usepackage{paralist} 
\usepackage{verbatim}
\usepackage{subfig} 
\usepackage{orcidlink}
\biboptions{comma,round,authoryear}
   
\begin{document}
\begin{frontmatter}
		\title{\textbf{A Conditional Quantile Approach to Vector-Valued Bivariate Lorenz Surfaces: Properties and Applications}}		
       \author{Shifna, P. R \orcidlink{0009-0002-4240-6941}\corref{cor1}}
	\ead{shifnarazack92@gmail.com}
	\author{S. M. Sunoj \orcidlink{0000-0002-6227-1506}}
	\ead{smsunoj@cusat.ac.in}

		
\address{Department of Statistics\\Cochin University of Science and Technology\\ Cochin 682 022, Kerala, INDIA.}

\begin{abstract}
The Lorenz curve is a fundamental tool for measuring inequality, but its extension to multivariate settings remains challenging due to the complex dependence structure among variables and the need to capture directional aspects of inequality. In this paper, we introduce a novel vector-valued bivariate Lorenz surface (VBLS) based on conditional distributions and conditional quantile functions. Unlike existing symmetric bivariate Lorenz surfaces, the proposed VBLS effectively captures the directional inequality arising from the conditional dependence between two variables. We establish several fundamental properties of the proposed surface and investigate its mathematical properties. The corresponding egalitarian surface is defined, leading to the development of associated vector-valued bivariate Gini measures for quantifying inequality. We further derive characterization results that demonstrate the uniqueness of the proposed VBLS within the underlying distributional framework. Nonparametric estimators of the VBLS are developed and their finite-sample performance is evaluated through a simulation study. The usefulness of the proposed methodology is also illustrated with applications to income inequality and actuarial data.
\end{abstract}
	\end{frontmatter}
\section{Introduction}
  
The Lorenz curve is one of the most important graphical and analytical tools for studying concentration, inequality and variability in non-negative distributions.   The Lorenz curve provides a cumulative representation of a distribution and serves as the basis for a variety of inequality measures, most notably the Gini index. Owing to its close connections with quantile functions, majorization theory, stochastic orders and concentration measures, the Lorenz curve has become a central object in both theoretical and applied statistics.
\par The Lorenz curve represents the distribution of income by plotting the cumulative percentage of total income received by the bottom $x\%$ of the population against the cumulative percentage of the population ranked by income. A perfectly equal distribution is represented by the 45-degree line (also known as the egalitarian line), in which each segment of the population receives an equal share of total income. For a finite population of $n$ individuals with ordered incomes $x_{1:n},x_{2:n},\ldots,x_{n:n}$, \citet{lorenz1905methods} defined the Lorenz curve  as
$$
L\!\left(\frac{i}{n}\right)
=
\frac{\sum_{j=1}^{i}x_{j:n}}
{\sum_{j=1}^{n}x_{j:n}},
$$
where the points $\left(\frac{i}{n},L\!\left(\frac{i}{n}\right)\right)$ are linearly interpolated to complete the curve. While this definition is discrete, \citet{gastwirth1971general} extended it to the continuous case using the quantile function. For a non-negative continuous random variable $X$ with cumulative distribution function $F (x) = P (X \leq x)$ and finite mean $\mu = E(X) = \int_{\mathcal{R}^+} {x\,dF(x)},$ the Lorenz curve is given by
$$
L(u)=\mu^{-1}\int_{0}^{u}Q(p)\,dp,
$$
where $Q(u)=F^{-1}(u)=\inf\{x:F(x)\ge u\},
\, 0\le u\le 1$ is the quantile function corresponding to $X$.  The Lorenz curve is a non-decreasing convex function lying below the 45-degree line, and its curvature reflects the degree of inequality. A popular measure derived from the Lorenz curve is the Gini index,
$$
G = 2\int_{0}^{1}[u-L(u)]\,du = 1-2\int_{0}^{1}L(u)\,du,
$$
which measures the area between the Lorenz curve and the egalitarian line.
 \par Although originally developed for the study of income distributions, the Lorenz curve has found applications far beyond economics. In actuarial science and insurance, Lorenz curves and Gini indices are widely used to evaluate risk classification systems, premium adequacy, claim concentration  and predictive performance of rate making models . Applications have also appeared in public health for measuring inequalities in health outcomes and health care utilization, in criminology for analysing crime concentration and in ecology for quantifying species abundance and resource allocation. See \citet{frees2011summarizing, frees2014insurance}, \citet{denuit2019model}, \citet{mauguen2016using}, \citet{christopoulos2017lorenz}, \citet{johnson2010brief}, \citet{bernasco2017more}, \citet{weiner1984meaning}, \citet{damgaard2000describing} for details. The broad applicability of Lorenz type measures stems from their ability to summarize how a total quantity is distributed across a population while preserving important distributional information that cannot be captured by simple location and dispersion measures. 
\par Many real-world problems involve several variables simultaneously. To address this complexity, researchers have proposed multivariate extensions of the Lorenz curve based on copulas, marginal quantiles, zonoids, optimal transport theory, etc. Some of these attempts include \citet{taguchi1972a, taguchi1972b}, \citet{arnold2014pareto, arnold1987majorization}, \citet{koshevoy1996lorenz}, \citet{sarabia2020lorenz}, \citet{grothe2022multivariate}, \citet{fan2022lorenz}, \citet{yordanov2025iterated}, representing certain important developments of Lorenz curve in multivariate framework. In recent years, substantial progress has also been made in extending quantile concepts to multivariate settings. The bivariate quantile framework developed by \citet{vineshkumar2019bivariate}, \citet{unnikrishnan2021properties}, \citet{pr2024multivariate}  and \citet{nair2025bivariate} provides a complete characterization of a bivariate distribution through marginal and conditional quantile functions. These developments offer a natural foundation for extending classical quantile-based tools, including Lorenz curves, to higher dimensions while retaining their probabilistic interpretation and analytical tractability. Despite these advances, most existing bivariate Lorenz curves treat the variables symmetrically, therefore, do not explicitly capture how inequality in one variable is distributed within subpopulations determined by another variable. This kind of directional information is often of substantial practical importance.  

\par Motivated by these considerations, the present study introduces a new vector-valued bivariate Lorenz surface (VBLS) based on conditional distributions and conditional quantile functions. The VBLS captures directional inequality that remains hidden under traditional symmetric approaches. The paper is organized as follows. Section 2 introduces the proposed vector-valued bivariate Lorenz surface and illustrates its examples. Section 3 establishes the main theoretical properties of the VBLS. The egalitarian surface and associated bivariate Gini measures are discussed in Section 4. Section 5 discusses the characterization of the VBLS. Estimation procedures  and simulation studies are presented in Section 6. Finally, Section 7 concludes the paper with the application of VBLS in two disciplines economics and actuarial science.

\section{Vector-valued bivariate  Lorenz surface}
Let $(X_1, X_2)$ be a non-negative absolutely continuous random vector with finite means. Let $F(x_1,x_2) = P(X_1 \leq x_1, X_2\leq x_2)$ be the joint distribution function, $\bar{F}(x_1,x_2)$ be the  joint survival function and marginal survival (distribution) functions, $\bar F_1 (x_1)(F_1 (x_1))$ and $\bar F_2(x_2)(F_2(x_2))$, respectively. Then we have 
\[\begin{aligned}
    \bar F(x_1,x_2)&=\bar F_1(x_1)\bar F_{21}(x_2\mid x_1),
\end{aligned}
\]
where $\bar F_{21}(x_2| x_1) = P(X_2 > x_2 \mid X_1 > x_1)$ is the conditional survival function of $X_2|X_1>x_1$. Correspondingly \citet{vineshkumar2019bivariate,unnikrishnan2021properties} defined the bivariate quantile function as $Q(u_1,u_2)=(Q_1(u_1),Q_{21}(u_1,u_2)),$ where $Q_1(u_1)=\inf\{x_1 : F_{1} (x_1) \ge u_1\}$ and $Q_{21}(u_1, u_2)=\inf\{x_2 : F_{21} (x_2,Q_1(u_1)) \ge u_2\}$. The joint distribution of $(X_1,X_2)$ can also be  characterized through the bivariate quantile function  given by $(Q_1(u_1), Q_{21}(u_1, u_2))$ or equivalently $(Q_{12}(u_1, u_2), Q_2(u_2))$. Moreover, the conditional quantile functions reduce to the corresponding marginal quantile functions as the conditioning level approaches zero, that is,
$\lim_{u_1\to0}Q_{21}(u_1,u_2)=Q_2(u_2),\,
\lim_{u_2\to0}Q_{12}(u_1,u_2)=Q_1(u_1).$
Motivated by the above conditional representation, for each fixed $u_2 \in [0,1)$, define
\[
F_{12}(x_1; u_2) = P\big(X_1 \le x_1 \mid X_2 > Q_2(u_2)\big),
\]
with mean
\[
\mu_{12}(u_2) = \int x_1 \, dF_{12}(x_1; u_2).
\]
Similarly, for each fixed $u_1 \in [0,1)$ define $F_{21}$ and $\mu_{21}$. Thus the quantile functions corresponding to the two conditional distributions are given by $Q_{12}$ and $Q_{21}$ respectively.
\begin{defi}
   The vector-valued bivariate Lorenz surface  (VBLS) for a random vector $(X_1, X_2)$ is defined as \[
\mathbf{L}(u_1, u_2) = \big( L_{12}(u_1, u_2), \, L_{21}(u_1, u_2) \big).
\] where\begin{equation}\label{lorenzdistribution_v}
\begin{aligned}
   & L_{12}(u_1, u_2) = \frac{1}{\mu_{12}(u_
   2)} 
\int_0^{Q_{12}(u_1, u_2)} x_1 \, dF_{12}(x_1, u_2),\\
&L_{21}(u_1, u_2) = \frac{1}{\mu_{21}(u_1)} 
\int_0^{Q_{21}(u_1, u_2)} x_2 \, dF_{21}(x_2, u_1).
\end{aligned}
\end{equation}
\end{defi}

The quantity $L_{12}(u_1,u_2)$ measures the cumulative share of $X_1$ by the lowest proportion $u_1$ of individuals within the subgroup whose $X_2$ exceeds its $u_2$-quantile. Likewise, $L_{21}(u_1,u_2)$ measures the cumulative share of $X_2$ contributed by the lowest proportion $u_2$ of individuals within the subgroup whose $X_1$ exceeds its $u_1$-quantile. Thus, each component of the VBLS is a Lorenz curve associated with a conditional distribution, quantifying the inequality of one variable within subpopulations determined by exceedance levels of the other variable.
\begin{ex}\label{eg_pareto}
Consider a bivariate Pareto distribution with survival function,
  \begin{equation}\label{paretosurvival}
      \bar F(x_1,x_2)=\left(\frac{x_1}{\theta_1}+\frac{x_2}{\theta_2}-1\right)^{-c},\quad x_1>\theta_1>0,\ x_2>\theta_2>0, \ c>1.
  \end{equation}
The marginal survival functions are given by 
$$
\bar F_i(x_i)=\left(\frac{x_i}{\theta_i}\right)^{-c}=1-u_i, \quad i=1,2$$
and 
$$\begin{aligned}
 & \bar F_{12}(x_1;u_2)=\frac{\bar F(x_1,x_2)}{\bar F_2(Q_2)}=\frac{\frac{x_1}{\theta_1}+\left( (1 - u_2)^{-\frac{1}{c}}   - 1 \right)^{-c}}{1 - u_2}, \\
 &\bar F_{21}(x_2;u_1)=\frac{\bar F(x_1,x_2)}{\bar F_1(Q_1)}=\frac{\left( (1 - u_1)^{-\frac{1}{c}} + \frac{x_2}{\theta_2} - 1 \right)^{-c}}{1 - u_1}.
\end{aligned}$$
Corresponding quantile functions are obtained as
\begin{equation*}
   Q_i(u_i) = (1-u_i)^{-1/c}\theta_i \quad i=1,2, 
\end{equation*}
and \begin{equation*}
\begin{aligned}
     &Q_{12}(u_1,u_2)=\left((1 -u_2)^{-1/
    c}\left((1 -u_1)^{-1/
    c} -1\right)+1 \right) \theta_1. \\
    &Q_{21}(u_1,u_2)=\left((1 -u_1)^{-1/
    c}\left((1 -u_2)^{-1/
    c} -1\right) +1\right) \theta_2. 
\end{aligned}
\end{equation*} 
\begin{equation*}
 \mu_{12}(u_2)= \frac{\left(\left(c - 1\right) \left(1 - u_{2}\right)^{\frac{1}{c}} + 1\right) {\theta}_{1}}{\left(c - 1\right) \left(1 - u_{2}\right)^{\frac{1}{c}}}  
\end{equation*}
\begin{equation*}
\begin{aligned}
    L_{12}(u_1,u_2)=&\frac{1}{\mu_{12}(u_2)}\int_{\theta_1}^{Q_{12}}x_1dF_{12}(x_1;u_2)\\
    =&\frac{
c(-1+(1-u_1)^{-1/c})(-1+u_1)+u_1-u_1(1-u_2)^{-1/c}+cu_1(1-u_2)^{-1/c}}{
1+(c-1)(1-u_2)^{-1/c}}.
\end{aligned}
\end{equation*}
Similarly
$$L_{21}(u_1,u_2)=\frac{
c(-1+(1-u_2)^{-1/c})(-1+u_2)+u_2-u_2(1-u_1)^{-1/c}+cu_2(1-u_1)^{-1/c}}{1+(c-1)(1-u_1)^{-1/c}}.$$
The figure \ref{pareto3dl12} and \ref{pareto3dl21} show the bivariate Lorenz surface for the Pareto distribution for $c=3$.
\begin{figure}[ht]
    \centering
    \includegraphics[scale=0.4]{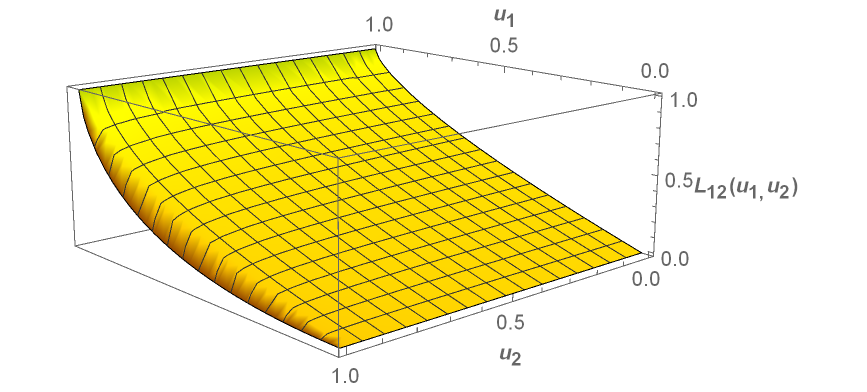}
    \caption{Plot of $L_{12}(u_1,u_2)$ for bivariate Pareto distribution}
    \label{pareto3dl12}
\end{figure}
\begin{figure}[ht]
    \centering
    \includegraphics[scale=0.4]{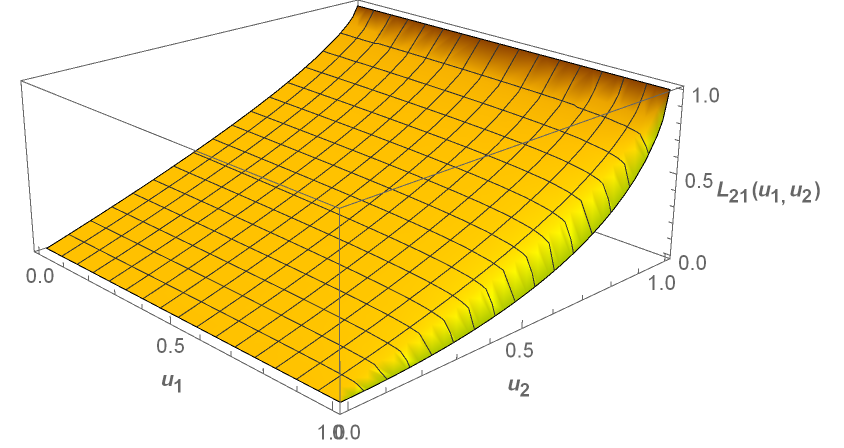}
    \caption{Plot of $L_{21}(u_1,u_2)$ for bivariate Pareto distribution}
    \label{pareto3dl21}
\end{figure}
\end{ex}
\par We now introduce an alternative formulation of the vector-valued bivariate Lorenz surface (VBLS) based on quantile functions. This representation is particularly useful in situations where the underlying distribution functions are not available in closed form, thereby offering a more analytically tractable approach. Moreover, the quantile-based framework accommodates a wide class of random variables, including discrete, continuous, and mixed types, making it suitable for applications in higher-dimensional settings. Applying the transformation $p=F_{12}$ in \eqref{lorenzdistribution_v},
we have $x_1=Q_{12}(p,u_2), \,
dF_{12}(x_1;u_2)=dp.$
Since $F_{12}(0;u_2)=0$ and $F_{12}\!\left(Q_{12}(u_1,u_2);u_2\right)=u_1,$
it follows that $L_{12}(u_1,u_2)
=\frac{1}{\mu_{12}(u_2)}
\int_{0}^{u_1}
Q_{12}(p,u_2)\,dp$. Similarly,  we obtain $L_{21}$.

\begin{defi}
   For a non-negative absolutely continuous random vector $(X_1,X_2)$ with finite means vector-valued bivariate Lorenz surface based on quantile functions is defined as
   \[
\mathbf{L}(u_1,u_2) = \bigl(L_{12}(u_1,u_2), L_{21}(u_1,u_2)\bigr),\quad 0\leq u_1,u_2<1
\]
where
\begin{equation}\label{lorenzquantile_v}
\begin{aligned}
   & L_{12}(u_1, u_2) = \frac{1}{\mu_{12}(u_2)} \int_0^{u_1} Q_{12}(p_, u_2)\, dp,\\
&L_{21}(u_1, u_2) = \frac{1}{\mu_{21}(u_1)} \int_0^{u_2} Q_{21}(u_1, p)\, dp.
\end{aligned}   
\end{equation}
\end{defi}
\begin{ex}
Consider a bivariate linear hazard function quantile model given in \citet{vineshkumar2019bivariate} with quantile functions
\[
Q_1(u_1)
=
\frac{1}{a+b}
\log\left(\frac{a + b u_1}{a(1-u_1)}\right),
\]
\[
Q_2(u_2)
=
\frac{1}{a+c}
\log\left(\frac{a + c u_2}{a(1-u_2)}\right),
\]
\[
Q_{12}(u_1,u_2)
=
\frac{1}{a+b+c u_2}
\log\left(
\frac{a + b u_1 + c u_2}{(a + c u_2)(1-u_1)}
\right),
\]
\[
Q_{21}(u_1,u_2)
=
\frac{1}{a+b u_1+c}
\log\left(
\frac{a + b u_1 + c u_2}{(a + b u_1)(1-u_2)}
\right).
\]
where $a>0,\; a+b>0,\; a+c>0$,\;$ a+b+c>0$ and $u_1,u_2 \in (0,1)$.  Then 
\[ \mu_{12}(u_2)
=\int_0^1 Q_{12}(p,u_2)\,dp=\frac{\log(1+\frac{b}{a+cu_2})}{b}.
\]
The corresponding component of the vector-valued bivariate Lorenz surface is given by
\[
\begin{aligned}
    L_{12}(u_1,u_2)=&
\frac{1}{\mu_{12}(u_2)}
\int_0^{u_1} Q_{12}(p,u_2)\,dp\\
=&\frac{(b-bu_1)\log(1-u_1)-(a+bu_1+cu_2)(\log(a+cu_2)-\log(a+bu_1+cu_2))}{(a+b+cu_2)\log(1+\frac{b}{a+cu_2})}.
\end{aligned}
\] Similarly \[
L_{21}(u_1,u_2)
=\frac{(c-cu_2)\log(1-u_2)-(a+bu_1+cu_2)(\log(a+bu_1)-\log(a+bu_1+cu_2))}{(a+c+bu_1)\log(1+\frac{c}{a+bu_1})}.
\]
The plots of VBLS for the bivariate linear hazard model for $a=b=c=1$ are shown in Figures \ref{l12linear} and \ref{l21linear}. 
\begin{figure}[ht]
    \centering
    \includegraphics[scale=0.4]{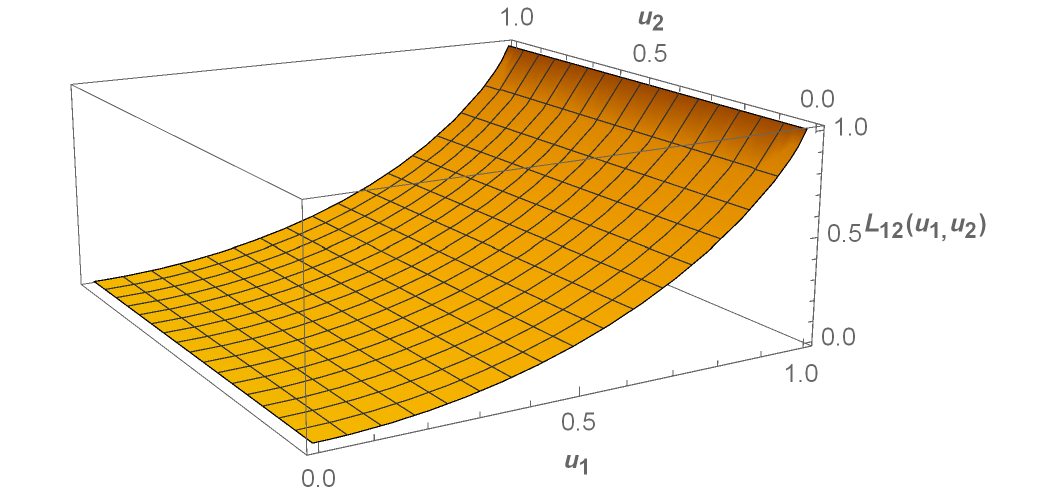}
    \caption{Graphical representation for $L_{12}(u_1,u_2)$ for bivariate linear hazard quantile model}
    \label{l12linear}
\end{figure}
\begin{figure}[ht]
    \centering
    \includegraphics[scale=0.4]{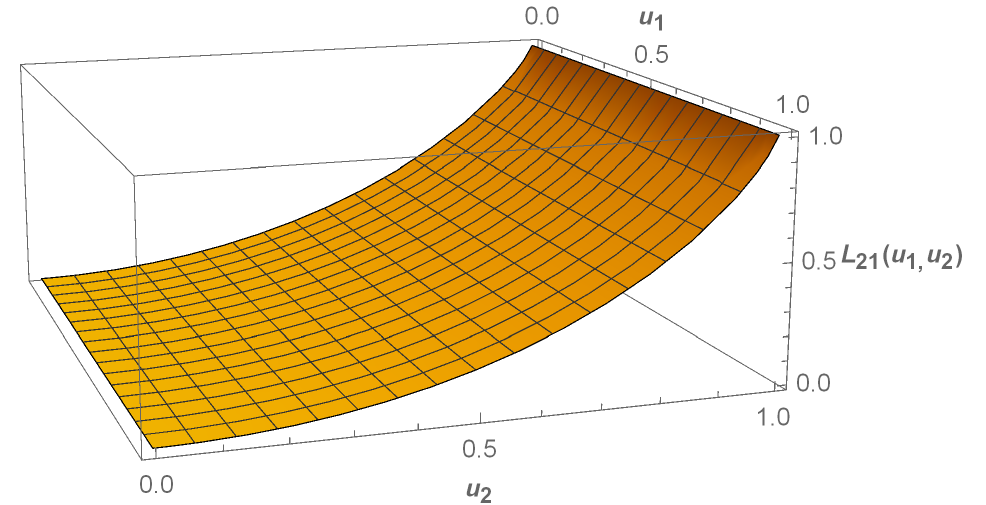}
    \caption{Graphical representation for $L_{21}(u_1,u_2)$ for  bivariate linear hazard quantile model}
    \label{l21linear} 
    \end{figure}

\end{ex}
To verify the equivalence of the distribution based and quantile based formulations of the VBLS, we illustrate using the following example the consistency through a bivariate model for which both the distribution functions and the corresponding quantile functions are available in closed form.
\begin{ex}
Consider the bivariate model with joint survival function 
$$\bar{F}(x_1,x_2)=(1+a_1x_1+a_2x_2)^{-p},\quad x_1,x_2>0,\,a_1,a_2>0,\,p>1.$$
We have $\bar{F}_i(x_i)=(1+a_ix_i)^{-p}$ and $Q_i(u_i)=\frac{(1-u_i)^{-1/p}-1}{a_i}$ for $i=1,2$. The conditional survival function
$$\bar F_{12}(x_1;u_2)=\frac{(a_1x_1+(1-u_2)^{-1/p})^{-p}}{(1-u_2)}$$ and $$f_{12}(x_1;u_2)=-\frac{\partial \bar F_{12}(x_1;u_2)}{\partial x_1}=\frac{pa_1(a_1x_1+(1-u_2)^{-1/p})^{-1-p}}{(1-u_2)}.$$ Thus
\[
Q_{12}(u_1, u_2) = (1-u_2)^{-1/p}\frac{((1-u_1)^{-1/p}-1) }{a_1}.
\]
Similarly
\[Q_{21}(u_1, u_2) =  (1-u_1)^{-1/p}\frac{((1-u_2)^{-1/p}-1) }{a_2}.
\]
We have $\mu_{12}(u_2)=\frac{1}{a_1(p-1)(1-u_2)^{1/p}}$
and hence using \eqref{lorenzdistribution_v}, we obtain
$$
    L_{12}(u_1,u_2)=\frac{1}{\mu_{12}(u_2)}\int_0^{Q_{12}}x_1f_{12}(x_1;u_2)dx_1 =p-(p-1)u_1-p(1-u_1)^{\frac{p-1}{p}}.
$$
Similarly
\[
L_{21}(u_1, u_2) 
= p-(p-1)u_2-p(1-u_2)^{\frac{p-1}{p}}.
\]
Using the definition of VBLS based on quantile function \eqref{lorenzquantile_v}, we have
$$
    L_{12}(u_1,u_2) = \frac{1}{\mu_{12}(u_2)}\int_0^{u_1}Q_{12}(p,u_2)dp = p-(p-1)u_1-p(1-u_1)^{\frac{p-1}{p}}.
$$
Similarly
\[
L_{21}(u_1, u_2) 
= p-(p-1)u_2-p(1-u_2)^{\frac{p-1}{p}}.
\]
This confirms the equivalence of the VBLS based on the distribution-based and quantile-based formulations, thereby illustrating the equivalence of their definitions. 
\end{ex}
\section{Properties of VBLS}
\begin{thm}
Let $(X_1,X_2)$ be a non-negative absolutely continuous 
random vector with finite means. Then the components of 
the VBLS satisfy the following boundary conditions for 
all $u_1, u_2 \in [0,1)$:
\begin{enumerate}[(i)]
    \item $L_{12}(0,u_2) = 0$ \text{ and } 
          $L_{21}(u_1,0) = 0$;
       \item $L_{12}(0,0) = L_{21}(0,0) = 0$;    
    
    \item $L_{12}(u_1,0) = L_1(u_1)$ and 
          $L_{21}(0,u_2) = L_2(u_2)$, where $L_1$ 
          and $L_2$ denote the univariate Lorenz curves 
          of $X_1$ and $X_2$ respectively;
          \item $\lim_{u_1\to 1^-} L_{12}(u_1,u_2) = 1$ 
          \text{ and } 
          $\lim_{u_2\to 1^-} L_{21}(u_1,u_2) = 1$;
   
    \item $\lim_{(u_1,u_2)\to(1,1)}
L_{12}(u_1,u_2)=1,$ and 
          $\lim_{(u_1,u_2)\to(1,1)}
L_{21}(u_1,u_2)=1$.
\end{enumerate}
\end{thm}
\begin{proof}
 We have \[
L_{12}(u_1,u_2)=\frac{1}{\mu_{12}(u_2)}\int_0^{u_1} Q_{12}(p,u_2)\,dp.
\] and\[
L_{21}(u_1,u_2)=\frac{1}{\mu_{21}(u_1)}\int_0^{u_2} Q_{21}(u_1,p)\,dp.
\]

We prove the stated properties for $L_{12}$; those for $L_{21}$ follow analogously.

\begin{enumerate}[(i)]

\item $L_{12}(0,u_2)=\frac{1}{\mu_{12}(u_2)}\int_0^{0} Q_{12}(p,u_2)\,dp=0.$
 \item $L_{12}(0,0)=\frac{1}{\mu_{12}(0)}\int_0^{0} Q_{12}(p,0)\,dp=0$.

\item By definition, $F_{12}(x_1;u_2) = P(X_1 \le x_1 \mid X_2 > Q_2(u_2)).$ Since $X_2$ is non-negative and absolutely continuous, we have $Q_2(0)=0$ and $P(X_2 > 0)=1$ at $u_2=0$. Thus,
\[\begin{aligned}
    F_{12}(x_1;0) &=P(X_1 \le x_1 \mid X_2 > Q_2(0))\\
    &=P(X_1 \le x_1 \mid X_2 > 0)\\
    &=P(X_1 \le x_1)=F_1(x_1)\\  
\end{aligned}
\]
Therefore, $Q_{12}(u_1,0)=Q_1(u_1)$ and $\mu_{12}(0)=\mu_1$. Consequently,
\[
L_{12}(u_1,0) = \frac{1}{\mu_{12}(0)}\int_0^{u_1}Q_{12}(p,0)\,dp = \frac{1}{\mu_1}\int_0^{u_1}Q_1(p)\,dp = L_1(u_1).
\]

\item As $u_1\to 1^-$, the numerator converges to $\int_0^1 Q_{12}(p,u_2)\,dp = \mu_{12}(u_2)$ by the dominated convergence theorem. Hence,
\[
\lim_{u_1\to 1^-} L_{12}(u_1,u_2) = \frac{\mu_{12}(u_2)}{\mu_{12}(u_2)} = 1.
\]
\item For any sequence $(u_1^{(n)}, u_2^{(n)}) \to (1,1)$, we have $u_1^{(n)} \to 1$. By part (iv), for each fixed $u_2^{(n)}$,
\[
\lim_{u_1^{(n)}\to 1^-} L_{12}(u_1^{(n)}, u_2^{(n)}) = 1.
\]
Since this holds for every such sequence, we conclude
\[
\lim_{(u_1,u_2)\to(1,1)} L_{12}(u_1,u_2) = 1.
\]
\end{enumerate}
Hence the stated boundary conditions hold.
\end{proof}

\begin{thm}\label{independence_thm}
 If $X_1$ and $X_2$ are independent, then
\[L_{12}(u_1, u_2) = L_1(u_1), 
\qquad 
L_{21}(u_1, u_2) = L_2(u_2).
\quad \forall u_1,u_2 \in[0,1) \]
\end{thm}

\begin{proof}
    Since $X_1$ and $X_2$ are independent, $P(X_1\le x_1,\;X_2>Q_2(u_2))
=
P(X_1\le x_1)\,P(X_2>Q_2(u_2)).$
Hence
$F_{12}(x_1; u_2) = P\big(X_1 \le x_1 \mid X_2> Q_2(u_2)\big) = F_1(x_1),$
$\mu_{12}(u_2) = \mu_1, \text{ and } Q_{12}(u_1, u_2) = Q_1(u_1).$
Hence
\[
L_{12}(u_1, u_2) = L_1(u_1).
\]
Similarly,
\[
L_{21}(u_1, u_2) = L_2(u_2). \]
\end{proof}

\begin{thm}
For each fixed $u_2\in[0,1)$, the function
$L_{12}(u_1,u_2)$ is non-decreasing in $u_1$ on $[0,1)$.
Similarly, for each fixed $u_1\in[0,1)$, the function
$L_{21}(u_1,u_2)$ is non-decreasing in $u_2$ on $[0,1)$.
\end{thm}

\begin{proof}
From the definition of the VBLS,
\[
L_{12}(u_1,u_2)
=
\frac{1}{\mu_{12}(u_2)}
\int_0^{u_1}Q_{12}(p,u_2)\,dp.
\]

For fixed $u_2$, differentiating with respect to $u_1$ 
\[
\frac{\partial}{\partial u_1}
L_{12}(u_1,u_2)
=
\frac{Q_{12}(u_1,u_2)}
{\mu_{12}(u_2)}.
\]

Since $X_1$ is non-negative, the conditional quantile function satisfies $Q_{12}(u_1,u_2)\ge 0,
\,
u_1,u_2\in[0,1).$ Moreover, $\mu_{12}(u_2)
=
\int_0^1Q_{12}(p,u_2)\,dp
>0.$ Therefore,

\[
\frac{\partial}{\partial u_1}
L_{12}(u_1,u_2)
\ge 0,
\qquad
u_1\in[0,1).
\]

Hence, $L_{12}(u_1,u_2)$ is non-decreasing in $u_1$ on $[0,1)$.  The proof for $L_{21}(u_1,u_2)$ follows similarly by differentiating with respect to $u_2$.
\end{proof}

\begin{thm}
Suppose that for each fixed $p\in[0,1)$, $\frac{Q_{12}(p,u_2)}{\mu_{12}(u_2)}$
is non-decreasing in $u_2$ on $[0,1)$. Then, for each fixed $u_1\in[0,1)$, $L_{12}(u_1,u_2)$ is non-decreasing in $u_2$ on $[0,1)$.  Similarly, if for each fixed $p\in[0,1)$, $\frac{Q_{21}(u_1,p)}{\mu_{21}(u_1)}$
is non-decreasing in $u_1$ on $[0,1)$, then for each fixed $u_2\in[0,1)$, $L_{21}(u_1,u_2)$ is non-decreasing in $u_1$ on $[0,1)$.
\end{thm}

\begin{thm}
For each fixed $u_2\in[0,1)$, $L_{12}(u_1,u_2)$ is convex in $u_1 \in [0,1)$. Similarly, for each fixed $u_1\in[0,1)$, $L_{21}(u_1,u_2)$ is convex in $u_2 \in [0,1)$.
\end{thm}
\begin{proof}
We prove the result for $L_{12}$, the proof of $L_{21}$ is analogous.

For each fixed $u_2$, $Q_{12}(p,u_2)$ is non-decreasing in $p$, since it is a quantile function. Since $L_{12}(u_1,u_2)$ is an integral function with non-negative integrand, it is absolutely continuous in $u_1$. Hence the derivative
\[
\frac{\partial}{\partial u_1}L_{12}(u_1,u_2) = \frac{Q_{12}(u_1,u_2)}{\mu_{12}(u_2)}
\]
exists almost everywhere and is non-decreasing in $u_1$, since $Q_{12}(\cdot,u_2)$ is non-decreasing and $\mu_{12}(u_2)>0$. A function that is absolutely continuous and has a non-decreasing derivative almost everywhere is convex. Therefore, $L_{12}(u_1,u_2)$ is convex in $u_1$.
\end{proof}
\begin{thm}
Suppose that for each fixed $p\in[0,1)$, the function $\frac{Q_{12}(p,u_2)}{\mu_{12}(u_2)}$ is convex in $u_2$ on $[0,1)$. Then, for each fixed $u_1\in[0,1)$, $L_{12}(u_1,u_2)$ is convex in $u_2$  on $[0,1)$.  Similarly, if for each fixed $p\in[0,1)$, the function $\frac{Q_{21}(u_1,p)}{\mu_{21}(u_1)}$ is convex in $u_1$ on $[0,1)$, then for each fixed $u_2\in[0,1)$, $L_{21}(u_1,u_2)$ is convex in $u_1$  on $[0,1)$.
\end{thm}

\begin{thm}\label{scale_thm}
Let $(X_1,X_2)$ be a non-negative absolutely continuous random vector with finite means, and let $Y_1=aX_1,\, Y_2=bX_2,$ where $a,b>0$. Then the VBLS of $(Y_1,Y_2)$ coincides with that of $(X_1,X_2)$; i.e., $\mathbf{L_Y}(u_1,u_2)=\mathbf{L_X}(u_1,u_2),
\qquad
\forall u_1,u_2\in[0,1).$
Equivalently, $L_{12}^{Y}(u_1,u_2)=L_{12}^{X}(u_1,u_2),
\,
L_{21}^{Y}(u_1,u_2)=L_{21}^{X}(u_1,u_2).$

\end{thm}

\begin{proof}
Since $Y_2=bX_2$, the marginal quantile function $Q_2^{Y}(u_2)$ satisfies $Q_2^{Y}(u_2)
=
\inf\{y_2:F_2^{Y}(y_2)\ge u_2\}
=
\inf\left\{y_2:F_2^{X}\!\left(\frac{y_2}{b}\right)\ge u_2\right\}.$ Let $x_2=\frac{y_2}{b}$, so that $y_2=bx_2$.  Hence, 
$$\begin{aligned}
  Q_2^{Y}(u_2)
&=\inf\{bx_2:F_2^{X}(x_2)\ge u_2\}\\
&=b\,\inf\{x_2:F_2^{X}(x_2)\ge u_2\}\\
&=b\,Q_2^{X}(u_2).  
\end{aligned}$$
Substituting $Y_1=aX_1$, $Y_2=bX_2$, and $Q_2^{Y}(u_2)=b\,Q_2^{X}(u_2)$, in $F_{12}^{Y}(y_1;u_2)$  we obtain,
$$\begin{aligned}
    F_{12}^{Y}(y_1;u_2)&=P(Y_1\le y_1\mid Y_2> Q_2^{Y}(u_2))\\
    &=P(aX_1\le y_1\mid bX_2> b\,Q_2^{X}(u_2))\\
    &=P\left(X_1\le \frac{y_1}{a}\mid X_2> Q_2^{X}(u_2)\right)\\
    &=F_{12}^{X}\left(\frac{y_1}{a};u_2\right).
\end{aligned}$$  

Therefore, $Q_{12}^{Y}(u_1,u_2)
=
\inf\left\{
y_1:
F_{12}^{Y}(y_1;u_2)\ge u_1
\right\}=\inf\left\{
y_1:
F_{12}^{X}\left(\frac{y_1}{a};u_2\right)\ge u_1
\right\}.$
Let $x_1=\frac{y_1}{a}$, so that $y_1=ax_1$. Then 
\[\begin{aligned}
   Q_{12}^{Y}(u_1,u_2)
&=\inf\left\{ax_1:F_{12}^{X}(x_1;u_2)\ge u_1
\right\}\\
&=a\,\inf\left\{x_1:F_{12}^{X}(x_1;u_2)\ge u_1
\right\}\\
&=a\,Q_{12}^{X}(u_1,u_2). 
\end{aligned}
\]
Further, $\mu_{12}^{Y}(u_2)
=
\int_0^1Q_{12}^{Y}(p,u_2)\,dp
=
a\int_0^1Q_{12}^{X}(p,u_2)\,dp
=
a\,\mu_{12}^{X}(u_2).$ Therefore,

\[\begin{aligned}
    L_{12}^{Y}(u_1,u_2)
&=\frac{1}{\mu_{12}^{Y}(u_2)}
\int_0^{u_1}Q_{12}^{Y}(p,u_2)\,dp\\
&=\frac{1}{a\,\mu_{12}^{X}(u_2)}
\int_0^{u_1}a\,Q_{12}^{X}(p,u_2)\,dp\\
&=L_{12}^{X}(u_1,u_2).
\end{aligned}
\]
The proof for $L_{21}$ follows by symmetry. Hence, the VBLS is scale-invariant.
\end{proof}
\section{Egalitarian surface and Gini index}
In the univariate case, the Lorenz curve $L(u)$ represents the proportion of total income held by the lowest $u$-fraction of the population. The egalitarian line is given by $L(u)=u,\, u\in[0,1)$, which represents perfect equality,  the situation in which every individual possesses exactly the same amount of the resource/income. Extending the notion of egalitarian line to the bivariate setting, we have the egalitarian surface for the VBLS as a benchmark where perfect equality holds in every conditional distribution, given in the following theorem.
\begin{thm}
Let $(X_1, X_2)$ be a non-negative absolutely continuous random vector with finite means. If there exist constants $c_1>0$ and $c_2>0$ such that
$X_1=c_1$ almost surely and $X_2=c_2$ almost surely, then for all $u_1,u_2\in[0,1)$,
\[
L_{12}(u_1,u_2)=u_1 \quad \text{and} \quad L_{21}(u_1,u_2)=u_2.
\]
Conversely, if
$L_{12}(u_1,u_2)=u_1 \text{ and }
L_{21}(u_1,u_2)=u_2 \text{ for all } u_1,u_2\in[0,1),$ then $X_1$ and $X_2$ are constants almost surely.
\end{thm}

\begin{proof}
We prove both directions of implication.
First, let $X_1=c_1$ almost surely. Then $P(X_1 = c_1) = 1.$ Now for any fixed $u_2 \in [0,1)$. Consider the event $A = \{X_2 > Q_2(u_2)\}$,
\[
P(X_1 = c_1 \mid A)
= \frac{P(X_1 = c_1 \cap A)}{P(A)}
= \frac{P(A)}{P(A)}
= 1.
\]
Hence, the conditional distribution of $X_1$ given $X_2>Q_2(u_2)$ is degenerate at $c_1$. Since $u_2$ was arbitrary, we have
$Q_{12}(u_1,u_2)=c_1 $ for all $ u_1,u_2\in[0,1)$
and $\mu_{12}(u_2)=\int_0^1 Q_{12}(p,u_2)\,dp=c_1.$
Thus,
\[
L_{12}(u_1,u_2) = \frac{1}{\mu_{12}(u_2)}\int_0^{u_1}{ Q_{12}(p, u_2)\,dp} = \frac{1}{c_1}\int_0^{u_1} c_1\,dp=u_1.
\]
Similarly, if $X_2=c_2$ almost surely, then for all $u_1,u_2\in[0,1)$, $L_{21}(u_1,u_2)=u_2.$

Conversely, assume $L_{12}(u_1,u_2)=u_1$ for all $u_1,u_2\in[0,1)$.  For any arbitrary $u_2\in[0,1)$, differentiating with respect to $u_1$
\[
\frac{\partial}{\partial u_1} L_{12}(u_1,u_2)
= \frac{Q_{12}(u_1,u_2)}{\mu_{12}(u_2)}
= 1.
\]
Thus, $Q_{12}(u_1,u_2)=\mu_{12}(u_2) \text{ for all } u_1\in[0,1)$,  which uniquely characterizes a degenerate distribution of the conditional distribution of $X_1$ given $X_2>Q_2(u_2)$. Hence, the conditional distribution is degenerate at $\mu_{12}(u_2)$. Now take $u_2=0$. Then $Q_2(0)$ is the lower endpoint of the support, so the conditioning event is the whole space. Hence, $P\big(X_1=\mu_{12}(0)\big)=1,$ $i.e.$, $X_1$ is constant almost surely. The same argument for $L_{21}$ shows that $X_2$ is constant almost surely.
\end{proof}
Based on the above theorem, we define the egalitarian surface for the VBLS of $(X_1, X_2)$ as follows.
\begin{defi}
The egalitarian surface for the vector-valued bivariate Lorenz surface is defined as a vector,
\[
\mathbf{L}(u_1,u_2) = (u_1,u_2) \quad \forall u_1,u_2\in[0,1).
\]
Equivalently, $L_{12}(u_1,u_2)=u_1 \quad \text{and} \quad L_{21}(u_1,u_2)=u_2 \quad \forall u_1,u_2\in[0,1).$
\end{defi}
The egalitarian surface for the VBLS implies that the joint distribution of $(X_1,X_2)$ attains the case of perfect equality. By the preceding theorem, this occurs if and only if the joint distribution is degenerate, that is, there exist constants $c_1>0$ and $c_2>0$ such that $X_1=c_1$ and $X_2=c_2$ almost surely. Hence, it corresponds to the complete absence of inequality in both variables. Any deviation from this surface indicates the presence of inequality in at least one of the variables.
\begin{thm}\label{bound_thm}
Let $(X_1,X_2)$ be a non-negative absolutely continuous random vector with finite means. Then for all
$u_1,u_2\in[0,1)$, $0\le L_{12}(u_1,u_2)\le u_1,
\text{ and }
0\le L_{21}(u_1,u_2)\le u_2.$

\end{thm}

\begin{proof}
First we prove $L_{12}(u_1,u_2)\ge 0$. Since $X_1$ is non-negative, all quantiles satisfy
$Q_{12}(p,u_2)\ge0$ for all $p\in[0,1)$. Hence $\int_0^{u_1}
Q_{12}(p,u_2)\,dp
\ge0$. Also, $\mu_{12}(u_2)=\int_0^1
Q_{12}(p,u_2)\,dp>0$.  Therefore,
$$\frac{1}{\mu_{12}(u_2)}\int_0^{u_1} Q_{12}(p,u_2)\,dp
\ge0
\implies
L_{12}(u_1,u_2)\ge0.
$$
To establish the upper bound, note that $Q_{12}(p,u_2)$ is a quantile function and is therefore non-decreasing in $p$. Hence, for every $0\le p\le u_1\le r\le 1,$ we have
$$
Q_{12}(p,u_2)\le Q_{12}(r,u_2).
$$
Integrating both sides with respect to $p$ over $[0,u_1]$ yields
$$
\int_0^{u_1}
Q_{12}(p,u_2)\,dp
\le
u_1\,Q_{12}(r,u_2),
\qquad r\in[u_1,1].
$$
Integrating this inequality again with respect to $r$ over $[u_1,1]$, we obtain
$$
(1-u_1)
\int_0^{u_1}
Q_{12}(p,u_2)\,dp
\le
u_1
\int_{u_1}^{1}
Q_{12}(r,u_2)\,dr.
$$
Adding $u_1
\int_0^{u_1}
Q_{12}(p,u_2)\,dp$ to both sides gives
$$
\int_0^{u_1}
Q_{12}(p,u_2)\,dp
\le
u_1
\int_0^{1}
Q_{12}(p,u_2)\,dp.
$$
Since $\mu_{12}(u_2)
=
\int_0^{1}
Q_{12}(p,u_2)\,dp,$ it follows that
$$
\int_0^{u_1}
Q_{12}(p,u_2)\,dp
\le
u_1\,\mu_{12}(u_2).
$$
Dividing both sides by $\mu_{12}(u_2)>0$, we obtain
$$
L_{12}(u_1,u_2)
=
\frac{1}{\mu_{12}(u_2)}
\int_0^{u_1}
Q_{12}(p,u_2)\,dp
\le
u_1.
$$
Therefore,$0\le L_{12}(u_1,u_2)\le u_1$.  Interchanging the roles of $X_1$ and $X_2$ yields $0\le L_{21}(u_1,u_2)\le u_2$.  Hence the result follows.
\end{proof}

\begin{thm}\label{Locshift_thm}
Let $(X_1,X_2)$ be a non-negative absolutely continuous random vector with finite means and define
$Y_1=X_1+c_1$, $Y_2=X_2+c_2$, where
$c_1,c_2\ge 0$. Then the VBLS of
$(Y_1,Y_2)$ satisfies
$$
L_{12}^{Y}(u_1,u_2)
=
\frac{\mu_{12}^{X}(u_2)}
{\mu_{12}^{X}(u_2)+c_1}
L_{12}^{X}(u_1,u_2)
+
\frac{c_1}
{\mu_{12}^{X}(u_2)+c_1}
u_1,
$$
$$
L_{21}^{Y}(u_1,u_2)
=
\frac{\mu_{21}^{X}(u_1)}
{\mu_{21}^{X}(u_1)+c_2}
L_{21}^{X}(u_1,u_2)
+
\frac{c_2}
{\mu_{21}^{X}(u_1)+c_2}
u_2.
$$
\end{thm}

\begin{proof}

Since
$Y_1=X_1+c_1$ and
$Y_2=X_2+c_2$,
$$\begin{aligned}F_2^Y(y_2)&=P(Y_2\le y_2)\\
&=P(X_2+c_2\le y_2)\\
&=P(X_2\le y_2-c_2)\\
&=F_2^X(y_2-c_2).
\end{aligned}$$
Let $x_2=y_2-c_2$. Then $y_2=x_2+c_2$ and hence quantile function of $Y_2$ becomes
\[\begin{aligned}
Q_2^Y(u_2) &= \inf\{y_2 : F_2^X(y_2 - c_2) \ge u_2\}\\
&= \inf\{x_2+c_2 : F_2^X(x_2) \ge u_2\}\\
&=c_2+\inf\{x_2 : F_2^X(x_2) \ge u_2\}\\
&= Q_2^X(u_2) + c_2.   
\end{aligned}
\]
Therefore,
$$
\begin{aligned}
 F_{12}^{Y}(y_1;u_2)&=P(Y_1\le y_1\mid Y_2>Q_2^{Y}(u_2))\\
 &=P(X_1+c_1\le y_1\mid X_2+c_2> Q_2^{Y}(u_2))
 \\
 &=P(X_1\le y_1-c_1\mid X_2> Q_2^{X}(u_2))\\
&=F_{12}(y_1-c_1;u_2).
\end{aligned}
$$
Substituting the expression for $F_{12}^{Y}$ in the  definition of the quantile function, 
$$
\begin{aligned}
Q_{12}^{Y}(u_1,u_2)&=\inf\{y_1: F_{12}^{Y}(y_1;u_2)\ge u_1\}\\
&=\inf\{y_1: F_{12}^{X}(y_1-c_1;u_2)\ge u_1\}.
\end{aligned}
$$
Let $x=y_1-c_1$. Then
$y_1=x+c_1$, and
$$
\begin{aligned}
 Q_{12}^{Y}(u_1,u_2)&=\inf\{x+c_1:
F_{12}^{X}(x;u_2)\ge u_1\}\\
&=c_1+\inf\{x:
F_{12}^{X}(x;u_2)\ge u_1\}\\
&=Q_{12}^{X}(u_1,u_2)+c_1.   
\end{aligned}
$$
The conditional mean transforms as
$$
\mu_{12}^{Y}(u_2)
=
\int_0^1
Q_{12}^{Y}(p,u_2)\,dp
=
\int_0^1
(Q_{12}^{X}(p,u_2)+c_1)\,dp
=
\mu_{12}^{X}(u_2)+c_1.
$$
Now, using the definition of the VBLS,
$$
L_{12}^{Y}(u_1,u_2)
=
\frac{1}{\mu_{12}^{Y}(u_2)}
\int_0^{u_1}
Q_{12}^{Y}(p,u_2)\,dp.
$$
Substituting the expressions for
$Q_{12}^{Y}$ and
$\mu_{12}^{Y}$,
$$
\begin{aligned}
    L_{12}^{Y}(u_1,u_2)
&=
\frac{1}{\mu_{12}^{X}(u_2)+c_1}
\int_0^{u_1}
(Q_{12}^{X}(p,u_2)+c_1)\,dp\\
&=\frac{
\int_0^{u_1}
Q_{12}^{X}(p,u_2)\,dp
+
c_1u_1
}
{\mu_{12}^{X}(u_2)+c_1}\\
&=\frac{
\mu_{12}^{X}(u_2)
L_{12}^{X}(u_1,u_2)
+
c_1u_1
}
{\mu_{12}^{X}(u_2)+c_1}.
\end{aligned}
$$
Thus $L_{12}^{Y}(u_1,u_2)
=
\frac{\mu_{12}^{X}(u_2)}
{\mu_{12}^{X}(u_2)+c_1}
L_{12}^{X}(u_1,u_2)
+
\frac{c_1}
{\mu_{12}^{X}(u_2)+c_1}
u_1.$
The proof for $L_{21}$ follows analogously by interchanging the roles of
$X_1$ and $X_2$.
\end{proof}
 
\begin{thm}
Let $(X_1,X_2)$ be a non-negative absolutely continuous random vector with finite means and define $Y_1=aX_1+c_1,\,
Y_2=bX_2+c_2,$ where $a,b>0$ and $c_1,c_2\ge0$. Then
\[
L_{12}^{Y}(u_1,u_2)
=
\frac{a\mu_{12}^{X}(u_2)}
     {a\mu_{12}^{X}(u_2)+c_1}
L_{12}^{X}(u_1,u_2)
+
\frac{c_1}
     {a\mu_{12}^{X}(u_2)+c_1}
u_1,
\]
and
\[
L_{21}^{Y}(u_1,u_2)
=
\frac{b\mu_{21}^{X}(u_1)}
     {b\mu_{21}^{X}(u_1)+c_2}
L_{21}^{X}(u_1,u_2)
+
\frac{c_2}
     {b\mu_{21}^{X}(u_1)+c_2}
u_2,
\]
for all $u_1,u_2\in[0,1)$.
\end{thm}
\begin{proof}
We prove the result for $L_{12}$.  Define $Z_1=aX_1,\, Z_2=bX_2.$ By the scale invariance property from Theorem \ref{scale_thm},
\begin{equation}\label{scale_eqn}
    L_{12}^{Z}(u_1,u_2)
=
L_{12}^{X}(u_1,u_2),
\qquad
\mu_{12}^{Z}(u_2)
=
a\,\mu_{12}^{X}(u_2).
\end{equation}
Now $Y_1=Z_1+c_1,\, Y_2=Z_2+c_2 $ and applying Theorem \ref{Locshift_thm} yields
\[
L_{12}^{Y}(u_1,u_2)
=
\frac{\mu_{12}^{Z}(u_2)}
     {\mu_{12}^{Z}(u_2)+c_1}
L_{12}^{Z}(u_1,u_2)
+
\frac{c_1}
     {\mu_{12}^{Z}(u_2)+c_1}
u_1.
\]
Substituting \eqref{scale_eqn} gives
\[
L_{12}^{Y}(u_1,u_2)
=
\frac{a\mu_{12}^{X}(u_2)}
     {a\mu_{12}^{X}(u_2)+c_1}
L_{12}^{X}(u_1,u_2)
+
\frac{c_1}
     {a\mu_{12}^{X}(u_2)+c_1}
u_1.
\]
The proof for $L_{21}$ follows analogously.
\end{proof}
The weights $\frac{a\mu_{12}^X(u_2)}
{a\mu_{12}^X(u_2)+c_1}$ and $\frac{c_1}
{a\mu_{12}^X(u_2)+c_1}$ sum to one, so the 
transformed Lorenz curve is a convex combination 
of the original Lorenz curve and the egalitarian 
line $u_1$, reflecting the equalizing effect of 
the location shift $c_1$. 
\par The classical Gini index is obtained from the area between the Lorenz curve and the egalitarian line and serves as a scalar measure of inequality. Since the proposed VBLS consists of two components, a natural extension is to define component-wise Gini measures associated with each direction of the conditional structure, denoted respectively by $G_{12}$ and $G_{21}$. These preserve the asymmetric information captured by the VBLS and quantify inequality within the joint system from both perspectives. The component $G_{12}$ measures the extent of inequality in $X_1$ through its conditional structure relative to $X_2$, whereas $G_{21}$ quantifies the corresponding inequality in $X_2$ through $X_1$. 
\begin{defi}[\textbf{Bivariate Gini Indices}]
The component-wise Gini indices associated with the VBLS are defined as
\[
G_{12}
=
1-2\int_0^1\int_0^1
L_{12}(u_1,u_2)\,du_1\,du_2,
\]
and
\[
G_{21}
=
1-2\int_0^1\int_0^1
L_{21}(u_1,u_2)\,du_1\,du_2.
\]

\end{defi}

The quantities $G_{12}$ and $G_{21}$ preserve directional information regarding inequality in the joint distribution. Under perfect equality, $L_{12}(u_1,u_2)=u_1,\, G_{12}=0.$
Under independence, $G_{12}=G_1$ and $G_{21}=G_2,$
where $G_1,G_2$ are the marginal Gini indices. In particular, differences between the two quantities indicate asymmetry in the conditional inequality structure. 

\begin{thm}
The component-wise Gini indices satisfy
\[
0\le G_{12},\,G_{21}\le 1.
\]
\end{thm}

\begin{proof}
We have from Theorem \ref{bound_thm},
\[
0\le L_{12}(u_1,u_2)\le u_1,
\qquad
0\le L_{21}(u_1,u_2)\le u_2.
\]
Integrating over $[0,1)^2$,
\[
0
\le
\int_0^1\int_0^1
L_{12}(u_1,u_2)\,du_1\,du_2
\le
\int_0^1\int_0^1
u_1\,du_1\,du_2.
\]
 It follows that
\[
0
\le
\int_0^1\int_0^1
L_{12}(u_1,u_2)\,du_1\,du_2
\le
\frac{1}{2}.
\]
Hence,
\[
0
\le
1-
2\int_0^1\int_0^1
L_{12}(u_1,u_2)\,du_1\,du_2
\le
1,
\]
which gives $0\le G_{12}\le1.$
Similarly, $0\le G_{21}\le1.$\\
\end{proof}
\section{Characterizations}
\begin{thm}
Let $(X_1,X_2)$ and $(Y_1,Y_2)$ be non-negative absolutely 
continuous random vectors with finite means. If
\[
L_{12}^{X}(u_1,u_2) = L_{12}^{Y}(u_1,u_2), \quad 
L_{21}^{X}(u_1,u_2) = L_{21}^{Y}(u_1,u_2) 
\qquad \forall\, u_1,u_2\in[0,1),
\]
and
\[
\mu_{12}^{X}(u_2) = \mu_{12}^{Y}(u_2)\; \forall\, u_2\in[0,1), 
\qquad 
\mu_{21}^{X}(u_1) = \mu_{21}^{Y}(u_1)\; \forall\, u_1\in[0,1),
\]
then $(X_1,X_2) \overset{d}{=} (Y_1,Y_2).$
\end{thm}

\begin{proof}
We have
\[
L_{12}^{X}(u_1,u_2)
=
\frac{1}{\mu_{12}^{X}(u_2)}
\int_0^{u_1} Q_{12}^{X}(p,u_2)\,dp.
\]
Since $L_{12}^X$ is an integral function, it is absolutely 
continuous in $u_1$. By the Lebesgue differentiation theorem,
\[
\frac{\partial}{\partial u_1}L_{12}^{X}(u_1,u_2)
=
\frac{Q_{12}^{X}(u_1,u_2)}{\mu_{12}^{X}(u_2)}
\quad \text{a.e.}
\]
Hence,
\[
Q_{12}^{X}(u_1,u_2)
=
\mu_{12}^{X}(u_2)\cdot
\frac{\partial}{\partial u_1}L_{12}^{X}(u_1,u_2)
\quad \text{a.e.}
\]
By assumption, $L_{12}^{X} = L_{12}^{Y}$ everywhere on 
$[0,1)^2$ and $\mu_{12}^{X}(u_2) = \mu_{12}^{Y}(u_2)$. 
Hence $Q_{12}^{X}(u_1,u_2) = Q_{12}^{Y}(u_1,u_2)$ almost 
everywhere in $u_1$. Since both are right-continuous quantile 
functions, equality holds everywhere:
\[
Q_{12}^{X}(u_1,u_2)
=
Q_{12}^{Y}(u_1,u_2)
\qquad \forall\, u_1,u_2\in[0,1).
\]
Similarly, from the equality of $L_{21}$ and $\mu_{21}$,
\[
Q_{21}^{X}(u_1,u_2)
=
Q_{21}^{Y}(u_1,u_2)
\qquad \forall\, u_1,u_2\in[0,1).
\]
When $u_2=0$, since $X_2$ is absolutely continuous, 
$F_2(Q_2(0))=1$, so the conditioning event 
$\{X_2> Q_2^X(0)\}$ has probability one. Therefore 
$Q_{12}^{X}(u_1,0) = Q_1^X(u_1)$, the marginal quantile 
of $X_1$. Similarly $Q_{21}^X(0,u_2) = Q_2^X(u_2)$. 
Thus the marginal quantile functions satisfy
\[
Q_1^{X}(u_1)
=Q_{12}^{X}(u_1,0)
=Q_{12}^{Y}(u_1,0)
=Q_1^{Y}(u_1),
\]
\[
Q_2^{X}(u_2)
=Q_{21}^{X}(0,u_2)
=Q_{21}^{Y}(0,u_2)
=Q_2^{Y}(u_2).
\]
Hence $X_1\overset{d}{=}Y_1$ and $X_2\overset{d}{=}Y_2$.
Since the joint distribution is uniquely characterized by 
the bivariate quantile functions $(Q_{12}(u_1,u_2),Q_2(u_2))$, 
or equivalently $(Q_1(u_1),Q_{21}(u_1,u_2))$, and both 
pairs agree for $(X_1,X_2)$ and $(Y_1,Y_2)$, we conclude
\[
(X_1,X_2)\overset{d}{=}(Y_1,Y_2).
\]
\end{proof}
The proof shows that the conditional quantile functions can be recovered from the VBLS through $Q_{12}(u_1,u_2)=\mu_{12}(u_2)\frac{\partial}{\partial u_1}L_{12}(u_1,u_2),$
and $Q_{21}(u_1,u_2)=\mu_{21}(u_1)\frac{\partial}{\partial u_2}L_{21}(u_1,u_2).$
Hence, the VBLS together with the conditional mean functions contains the complete conditional quantile structure of the underlying distribution.

\begin{thm}
Let $(X_1,X_2)$  be a non-negative absolutely continuous random vector  with finite means. Then  $(L_{12}(u_1,u_2),L_{21}(u_1,u_2))=(u_1^{\alpha_1},u_2^{\alpha_2})$ and $\mu_{12}(u_2)= m_1,\,\mu_{21}(u_1)= m_2,$ where $\alpha_1,\alpha_2>1$ and $m_1,m_2>0$ if and only if  $(X_1,X_2)$ has the independent bivariate power distribution with joint distribution
\[
F(x_1,x_2)=\left(\frac{x_1}{\alpha_1 m_1}\right)^{\frac{1}{\alpha_1-1}}\left(\frac{x_2}{\alpha_2 m_2}\right)^{\frac{1}{\alpha_2-1}}, \; 0\le x_1\le \alpha_1 m_1, \; 0\le x_2\le \alpha_2 m_2.
\]
\end{thm}

\begin{proof}
Let  $(L_{12}(u_1,u_2),L_{21}(u_1,u_2))=(u_1^{\alpha_1},u_2^{\alpha_2})$ and $\mu_{12}(u_2)= m_1,\,\mu_{21}(u_1)= m_2,$ then  
\[\begin{aligned}
   Q_{12}(u_1,u_2)&=
\mu_{12}(u_2)
\frac{\partial}{\partial u_1}
L_{12}(u_1,u_2),\\
&=m_1\frac{\partial}{\partial u_1}(u_1^{\alpha_1})
=
\alpha_1 m_1u_1^{\alpha_1-1}.
\end{aligned}\]
Similarly, $Q_{21}(u_1,u_2)
=
\alpha_2 m_2u_2^{\alpha_2-1}.$

Thus the conditional quantile functions are independent of the conditioning levels $u_2$ and $u_1$, respectively. Hence by Theorem \ref{independence_thm}, $Q_1(u_1)=\alpha_1 m_1u_1^{\alpha_1-1}, \text{ and } Q_2(u_2)=\alpha_2 m_2u_2^{\alpha_2-1}.$
The corresponding bivariate distribution function is bivariate power with marginals $F_i(x_i)
=
\left(\frac{x_i}{\alpha_i m_i}\right)^{\frac{1}{\alpha_i-1}}, \; 0\le x_i\le \alpha_i m_i, \; i=1,2.$ 

Conversely, if $(X_1,X_2)$ has independent bivariate power distribution with marginals $F_i(x_i)
=
\left(\frac{x_i}{\alpha_i m_i}\right)^{\frac{1}{\alpha_i-1}}, \; 0\le x_i\le \alpha_i m_i, \; i=1,2.$, then
\[
Q_i(u_i)=\alpha_i m_i u_i^{\alpha_i-1},
\qquad i=1,2.
\]
Since $X_1$ and $X_2$ are independent, $Q_{12}(u_1,u_2)=Q_1(u_1),\,Q_{21}(u_1,u_2)=Q_2(u_2),$
and $\mu_{12}(u_2)=m_1,\,
\mu_{21}(u_1)=m_2.$

Hence,
\[
L_{12}(u_1,u_2)
=
\frac1{m_1}\int_0^{u_1}\alpha_1m_1p^{\alpha_1-1}\,dp
=
u_1^{\alpha_1},
\]
and similarly,
\[
L_{21}(u_1,u_2)=u_2^{\alpha_2}.
\]
This completes the proof.
\end{proof}
\begin{thm}
Let $(X_1,X_2)$ be a non-negative absolutely continuous random vector  with finite means. Then $(L_{12}(u_1,u_2),L_{21}(u_1,u_2))=(1-(1-u_1)^{\gamma_1},1-(1-u_2)^{\gamma_2})$ and  $\mu_{12}(u_2)=m_1,\,\mu_{21}(u_1)=m_2,$ where \(0<\gamma_1,\gamma_2<1\) and $m_1,m_2>0$ if and only if $(X_1,X_2)$ has the independent bivariate Pareto distribution with joint survival function
\[
\bar F(x_1,x_2)
=
\left((\frac{m_1\gamma_1}{x_1})^{\frac{1}{1-\gamma_1}}\right)
\left((\frac{m_2\gamma_2}{x_2})^{\frac{1}{1-\gamma_2}}\right), \, x_1\ge m_1\gamma_1,\,
x_2\ge m_2\gamma_2.
\]
\end{thm}
\begin{proof}
Let $L_{12}(u_1,u_2)=1-(1-u_1)^{\gamma_1}$ and  $\mu_{12}(u_2)=m_1,$ we have
$$
Q_{12}(u_1,u_2)
=
\mu_{12}(u_2)
\frac{\partial}{\partial u_1}
L_{12}(u_1,u_2)
=
m_1\gamma_1(1-u_1)^{\gamma_1-1}.
$$
Similarly, $Q_{21}(u_1,u_2)= m_2\gamma_2(1-u_2)^{\gamma_2-1}.$

Since neither depends on the conditioning level, independence follows exactly as in Theorem \ref{independence_thm} and hence $Q_i(u_i)=m_i\gamma_i(1-u_i)^{\gamma_i-1},\, i=1,2$. Solving for $u_i$ by setting  $x_i=Q_i(u_i)$, $$\bar F_i(x_i)=
\left(
\frac{m_i\gamma_i}{x_i}
\right)^{\frac{1}{1-\gamma_i}},\,x_i\ge m_i \gamma_i, \, i=1,2,$$ which is the Pareto distribution function.  Hence the bivariate distribution corresponds to independent bivariate  Pareto distribution.

Conversely, let\[
\bar F(x_1,x_2)
=
\left((\frac{m_1\gamma_1}{x_1})^{\frac{1}{1-\gamma_1}}\right)
\left((\frac{m_2\gamma_2}{x_2})^{\frac{1}{1-\gamma_2}}\right), \, x_1\ge m_1\gamma_1,\,
x_2\ge m_2\gamma_2.
\]
Since $X_1$ and $X_2$ are independent, we have $\bar F_i(x_i)=
\left(
\frac{m_i\gamma_i}{x_i}
\right)^{\frac{1}{1-\gamma_i}}$ and corresponding quantile $Q_i(u_i)=m_i\gamma_i(1-u_i)^{\gamma_i-1},\, i=1,2$. Then $\mu_{12}(u_2)=m_1,\text{ and }
\mu_{21}(u_1)=m_2.$ Hence \[
L_{12}(u_1,u_2)
=
\frac1{m_1}\int_0^{u_1}m_1\gamma_1(1-p)^{\gamma_1-1}\,dp
=
1-(1-u_1)^{\gamma_1},
\]
and similarly,
\[
L_{21}(u_1,u_2)=1-(1-u_2)^{\gamma_2}.
\] Hence the result follows.
\end{proof}
\begin{thm}
Let $(X_1,X_2)$ be a non-negative absolutely continuous random vector. Suppose that $(L_{12}(u_1,u_2),L_{21}(u_1,u_2))=(u_1+(1-u_1)\log(1-u_1),u_2+(1-u_2)\log(1-u_2))$, $\mu_{12}(u_2)
=\frac{{\lambda}_2}{{\lambda}_1 {\lambda}_2 - \ln(1 - u_2) \theta}$ and $\mu_{21}(u_1)=\frac{{\lambda}_1}{{\lambda}_2 {\lambda}_1 - \ln(1 - u_1) \theta},$ where $\lambda_1,\lambda_2>0$ and $0\le\theta\le\lambda_1\lambda_2$.  Then
\[Q_{12}(u_1,u_2)=\frac{-\lambda_2\log(1-u_1)}{\lambda_1\lambda_2-\theta\log(1-u_2)}\quad 
Q_{21}(u_1,u_2)=\frac{-\lambda_1\log(1-u_2)}{\lambda_2\lambda_1-\theta\log(1-u_1)},
\]
and consequently $(X_1,X_2)$ follows Gumbel's bivariate exponential distribution.
\end{thm}
\begin{proof}
Since $L_{12}(u_1,u_2) =u_1+(1-u_1)\log(1-u_1),$ we have
\[
\frac{\partial}{\partial u_1}
L_{12}(u_1,u_2)=-\log(1-u_1).
\]
Also $\mu_{12}(u_2)
=\frac{{\lambda}_2}{{\lambda}_1 {\lambda}_2 - \ln(1 - u_2) \theta}$.  Therefore,
\[
Q_{12}(u_1,u_2)=
\frac{-\lambda_2\log(1-u_1)}{\lambda_1\lambda_2-\theta\log(1-u_2)}.
\]
The marginal quantile function can be recovered from $Q_{12}(u_1,u_2)$ as $Q_{12}(u_1,0)= Q_1(u_1)=\frac{-\log(1-u_1)}{\lambda_1}$
$\bar{F}_1(x_1)=\exp({-\lambda_1x_1})$, $\bar{F}_{12}(x_1|x_2)=\exp{(-(\lambda_1+\theta x_2)x_1)}$ and  hence $$\bar{F}(x_1,x_2)=\exp({-\lambda_1x_1-\lambda_2x_2-\theta x_1x_2}),\,x_1,x_2\ge 0$$ which is  Gumbel's bivariate exponential distribution.
\end{proof}
\section{Nonparametric Estimation of the VBLS}
Let $(X_{1i},X_{2i})_{i=1}^n$ be an i.i.d. sample from a non-negative absolutely continuous bivariate distribution with finite first moments. In this section, we develop a nonparametric estimator for the VBLS, $\mathbf{L}(u_1,u_2)$.
\par Fix $u_2\in(0,1)$. Let $X_{2(1)}\le X_{2(2)}\le\cdots\le X_{2(n)}$ be the order statistics of \(X_2\). The empirical marginal quantile is defined by $\hat Q_2(u_2)=X_{2(\lceil nu_2\rceil)}$. Define the  subsample $S_{u_2}=\{i:X_{2i}> \hat Q_2(u_2)\},$ and let $n_{u_2}$ be its sample size. Let the order statistics of $\{X_{1i}:i\in S_{u_2}\}$ be $X_{1(1)}^{u_2}\le X_{1(2)}^{u_2} \le \cdots \le X_{1(n_{u_2})}^{u_2}.$ The empirical conditional quantile function of \(X_1\) given
\(X_2> \hat Q_2(u_2)\) is defined by $\hat Q_{12}(p,u_2) = X_{1(\lceil n_{u_2}p\rceil)}^{u_2},
\,0<p<1.$ The empirical conditional mean is estimated by $\hat\mu_{12}(u_2)
=
\frac1{n_{u_2}}
\sum_{i\in S_{u_2}}
X_{1i}.$ Since the empirical quantile function is stepwise constant, let $k=
\lfloor n_{u_2}u_1\rfloor, \; 0<u_1<1$.  Then
\[
\hat L_{12}(u_1,u_2)
=
\frac1{\hat\mu_{12}(u_2)}
\left[
\frac1{n_{u_2}}
\sum_{i=1}^{k}
X_{1(i)}^{u_2}
+
\left(
u_1-\frac{k}{n_{u_2}}
\right)
X_{1(k+1)}^{u_2}
\right].
\]
The estimator \(\hat L_{21}(u_1,u_2)\) is obtained analogously by interchanging the roles of \(X_1\) and \(X_2\). Specifically, for fixed \(u_1\in(0,1)\), compute the empirical quantile \(\hat Q_1(u_1)\), extract the subsample $\{X_{2i}:X_{1i}>\hat Q_1(u_1)\}$ construct the corresponding empirical conditional quantile and conditional mean, and proceed as above.

\par The proposed nonparametric estimator of the VBLS is $$\mathbf{\hat L}(u_1,u_2)=\Big(\hat L_{12}(u_1,u_2),\hat L_{21}(u_1,u_2)\Big), \; u_1,u_2\in(0,1).$$
The estimator preserves the conditional structure inherent in the VBLS and naturally extends the classical empirical Lorenz curve to the vector-valued bivariate setting. 
\begin{thm}\label{VBLS_consistency}
Let $\{(X_{1i},X_{2i})\}_{i=1}^n$ be an i.i.d. sample from a bivariate distribution with absolutely continuous distribution function $F$. Assume that $E|X_1|<\infty$, $E|X_2|<\infty$, and $\mu_{12}(u_2)>0,\,
\mu_{21}(u_1)>0,$ for all $(u_1,u_2)\in (0,1)^2$.
Then the nonparametric estimator $\mathbf{\hat L}(u_1,u_2)$  converges almost surely to the VBLS $\mathbf{L}(u_1,u_2)$. Equivalently, $\hat L_{12}(u_1,u_2)
\xrightarrow{a.s.}
L_{12}(u_1,u_2),$ and $\hat L_{21}(u_1,u_2)
\xrightarrow{a.s.}
L_{21}(u_1,u_2).$
\end{thm}

\begin{proof}

Using the Glivenko--Cantelli theorem, empirical distribution $\hat F_2$ converges uniformly almost surely to $F_2$. Also $F_2$ is continuous, empirical quantiles $\hat Q_2(u_2)
\xrightarrow{a.s.}
Q_2(u_2).$ Consider the empirical conditional distribution function

\[
\hat F_{12}(x_1|u_2)
=
\frac{
\sum_{i=1}^{n}
1\{X_{1i}\le x_1,\,
X_{2i}>\hat Q_2(u_2)\}
}{
\sum_{i=1}^{n}
1\{X_{2i}>\hat Q_2(u_2)\}
}.
\]
By the law of large numbers the numerator converges almost surely to $P(X_1\le x_1,\,
X_2> Q_2(u_2)),$ while the denominator converges almost surely to $P(X_2>Q_2(u_2))=1-u_2>0.$
Hence, $\hat F_{12}(\cdot|u_2)
\xrightarrow{a.s.}
F_{12}(\cdot|u_2).$
Uniform convergence of distribution functions implies convergence of the corresponding conditional quantile functions $\hat Q_{12}(p,u_2)
\xrightarrow{a.s.}
Q_{12}(p,u_2), \,  p\in(0,1)$.  Hence, by the dominated convergence theorem, $\int_0^{u_1}
\hat Q_{12}(p,u_2)\,dp
\xrightarrow{a.s.}
\int_0^{u_1}
Q_{12}(p,u_2)\,dp$.  Further, by the law of large numbers, the empirical conditional mean $\hat\mu_{12}(u_2)\xrightarrow{a.s.}\mu_{12}(u_2)$.  Therefore, by the continuous mapping theorem, $\hat{L}_{12}(u_1,u_2)
\xrightarrow{a.s.} L_{12}(u_1,u_2)$.  The proof for $\hat L_{21}(u_1,u_2)$
follows similarly by interchanging the roles of $X_1$ and $X_2$. Hence, $\mathbf{\hat L}(u_1,u_2)
\xrightarrow{a.s.}
\mathbf{L}(u_1,u_2).$
\end{proof}

\subsection{Simulation Study}
To examine the finite-sample behavior of the proposed nonparametric estimator of the VBLS, a Monte Carlo simulation study was conducted under the bivariate Pareto model introduced in the Example \ref{eg_pareto}. The distribution was chosen due to its ability to capture heavy-tailed characteristics commonly observed in income and wealth data. The estimator was evaluated at the probability combinations $(u_1,u_2)=(0.3,0.3),(0.5,0.5)$, and $(0.7,0.7)$, corresponding to lower, central, and upper regions of the conditional distribution. For each setting, $500$ Monte Carlo replications were generated for sample sizes $n=50,100,200$, and $500$. The resulting estimates of $\hat{L}_{12}(u_1,u_2)$ and $\hat{L}_{21}(u_1,u_2)$ were assessed using absolute bias and mean squared error (MSE). The simulation results are summarized in Table \ref{simulation_VBLS}. 
\begin{table}[ht]
\centering
\renewcommand{\arraystretch}{1.3}
\setlength{\tabcolsep}{10pt}
\caption{Estimates, Absolute Bias and MSE of $(\hat L_{12}, \hat L_{21})$ based on simulated data}
\label{simulation_VBLS}
\resizebox{0.6\linewidth}{!}{%
\begin{tabular}{ccccc}
\toprule
$(u_1,u_2)$ & $n$ & $(\hat L_{12}, \hat L_{21})$ & Bias & MSE\\
\midrule

(0.3,0.3) & 50  & (0.2071, 0.2088) & (0.0027, 0.0042) & (0.0003, 0.0003)\\
          & 100 & (0.2047, 0.2058) & (0.0011, 0.0022) & (0.0002, 0.0002)\\
          & 200 & (0.2047, 0.2054) & (0.0004, 0.0008) & (0.0001, 0.0001)\\
          & 500 & (0.2049, 0.2047) & (0.0005, 0.0004) & (0.0000, 0.0000)\\

\addlinespace

(0.5,0.5) & 50  & (0.3583, 0.3566) & (0.0090, 0.0064) & (0.0014, 0.0013)\\
          & 100 & (0.3526, 0.3537) & (0.0025, 0.0015) & (0.0007, 0.0007)\\
          & 200 & (0.3507, 0.3518) & (0.0027, 0.0020) & (0.0005, 0.0004)\\
          & 500 & (0.3497, 0.3497) & (0.0010, 0.0019) & (0.0002, 0.0002)\\

\addlinespace

(0.7,0.7) & 50  & (0.5303, 0.5322) & (0.0185, 0.0249) & (0.0046, 0.0047)\\
          & 100 & (0.5210, 0.5233) & (0.0106, 0.0175) & (0.0026, 0.0027)\\
          & 200 & (0.5152, 0.5150) & (0.0045, 0.0022) & (0.0014, 0.0016)\\
          & 500 & (0.5111, 0.5118) & (0.0007, 0.0048) & (0.0006, 0.0009)\\

\bottomrule
\end{tabular}}
\end{table}
It can be seen that the estimation error decreases steadily as the sample size increases across all combinations of $(u_1,u_2)$. In particular, both the absolute bias and MSE exhibit a declining pattern, indicating improved estimator accuracy in larger samples. Similar behavior is observed for both components of the VBLS, suggesting stable performance of the proposed estimator and providing empirical support for the consistency result established in Theorem \ref{VBLS_consistency}.
\section{Data analysis}
\subsection{Vietnam household survey data}
To demonstrate the applicability of the proposed VBLS in the analysis of socioeconomic inequality, we consider the \texttt{VietNamH} dataset available in the \texttt{Ecdat} package. The dataset contains information on household expenditures collected from a nationally representative survey of Vietnamese households. The variable $X_1$ represents the total household expenditure  that reflects the overall household consumption and $X_2$ represents the medical expenditure $i.e.$, spending on health related needs. These variables are of particular interest because health expenditures often exhibit substantially greater concentration than general household consumption and may therefore reveal important dimensions of economic inequality.
 
The proposed nonparametric estimator was employed to estimate the  Lorenz surfaces $L_{12}(u_1,u_2)$ and $L_{21}(u_1,u_2)$ over a grid of probability values in the unit square. Selected values of the estimated VBLS are reported in Table \ref{vietnam_VBLS}.

\begin{table}[ht]
\centering
\begin{tabular}{ccc}
\hline
$(u_1,u_2)$ &
$\hat{L}_{12}(u_1,u_2)$ &
$\hat{L}_{21}(u_1,u_2)$ \\
\hline
$(0.25,0.25)$ & 0.0915 & 0.0026 \\
$(0.50,0.50)$ & 0.2508 & 0.0374 \\
$(0.75,0.75)$ & 0.4853 & 0.1563 \\
$(0.90,0.90)$ & 0.7052 & 0.3753 \\
\hline
\end{tabular}
\caption{Selected values of the estimated VBLS for the Vietnam expenditure data}
\label{vietnam_VBLS}
\end{table}
The values increase with the probability levels, reflecting the cumulative nature of the Lorenz framework. Moreover, at every selected point, $L_{12}(u_1,u_2)>L_{21}(u_1,u_2)$, indicating a pronounced asymmetry in the concentration structure of the two variables. The estimated surfaces of $L_{12}(u_1,u_2)$ and $L_{21}(u_1,u_2)$ are displayed in Figures \ref{3d_exp_L12} and \ref{3d_exp_L21} respectively. To facilitate interpretation, diagonal sections of the estimated surfaces were examined by considering $L_{12}(u,u) \text{ and } 
L_{21}(u,u),\, 0\le u\le 1$. The resulting curves are shown in Figure \ref{diag_exp}. Both curves satisfy the standard Lorenz properties. However, the curve corresponding to $L_{12}(u,u)$ lies consistently above that of $L_{21}(u,u)$, demonstrating that cumulative shares of total expenditure accumulate more rapidly than cumulative shares of medical expenditure when equal probability levels are considered. The gap between the two curves widens substantially in the middle and upper portions of the distribution, indicating increasing divergence in concentration patterns as larger expenditure shares are accumulated.
\begin{figure}[ht]
    \centering
    \includegraphics[scale=0.4]{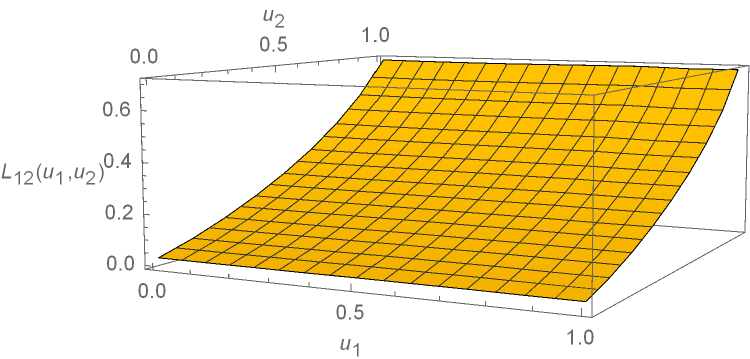}
    \caption{Estimated VBLS surface $\hat{L}_{12}(u_1,u_2)$ for the Vietnam expenditure data}
    \label{3d_exp_L12}
\end{figure}
\begin{figure}[ht]
    \centering
    \includegraphics[scale=0.4]{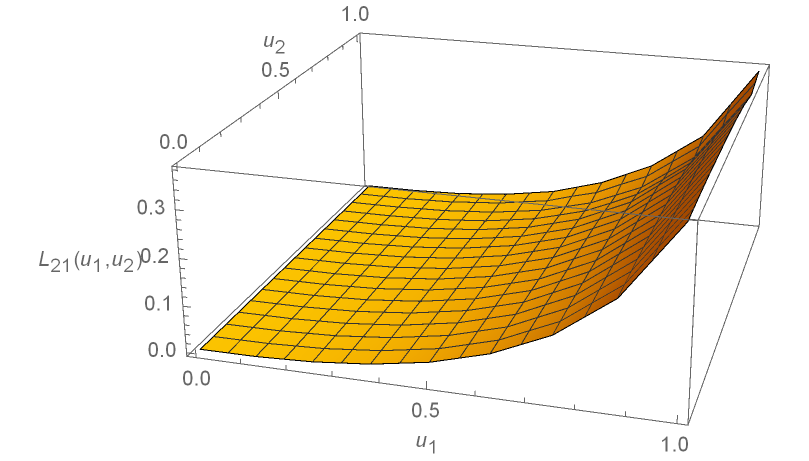}
    \caption{Estimated VBLS surface $\hat{L}_{21}(u_1,u_2)$ for the Vietnam expenditure data }
    \label{3d_exp_L21}
\end{figure}
\begin{figure}[ht]
    \centering
    \includegraphics[scale=0.4]{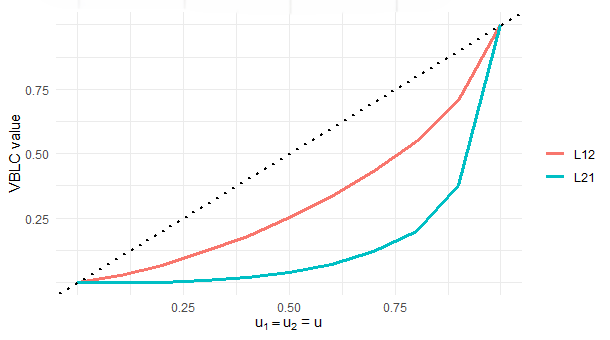}
    \caption{Diagonal comparison of the estimated VBLSs for the Vietnam expenditure data}
    \label{diag_exp}
\end{figure}
Using the proposed double integral formulation, the estimated bivariate Gini indices are $G_{12}=0.5215$ and $G_{21}=0.8488$. Since $G_{12}<G_{21}$, the average level of the surface $L_{12}$ exceeds that of $L_{21}$, indicating that inequality associated with medical expenditure is substantially greater than inequality associated with total household expenditure when the dependence structure between the variables is taken into account. These findings highlight the usefulness of the proposed VBLS as a tool for studying multidimensional inequality and for uncovering directional features that remain hidden under traditional approaches.

\subsection{Workers compensation data}
To further demonstrate the versatility of the proposed VBLS framework, we turn to an application in actuarial science using the Workers Compensation dataset available in the \texttt{insuranceData} package. The dataset contains 847 observations corresponding to workers' compensation insurance contracts. Two variables $X_1$ and $X_2$ represent the premium or amount collected by the insurer from policyholders and the loss amount or the claim payments made by the insurer, respectively.

The surfaces $(L_{12}(u_1,u_2),L_{21}(u_1,u_2))$ over $(u_1,u_2)$ were estimated using the proposed nonparametric estimator . 
Table \ref{actuarial_VBLS} reports selected values of the estimated VBLS.

\begin{table}[ht]
\centering
\caption{Selected values of the estimated VBLS for the Workers Compensation data}
\label{actuarial_VBLS}
\begin{tabular}{ccc}
\hline
$(u_1,u_2)$ & $\hat L_{12}(u_1,u_2)$ & $\hat L_{21}(u_1,u_2)$ \\
\hline
$(0.25,0.25)$ & 0.0146 & 0.0192 \\
$(0.50,0.50)$ & 0.0584 & 0.1055 \\
$(0.75,0.75)$ & 0.3160 & 0.4667 \\
$(0.90,0.90)$ & 0.5997 & 0.7437 \\
\hline
\end{tabular}
\end{table}

The values increase steadily with $u_1$ and $u_2$, reflecting the cumulative nature of the Lorenz framework. The consistent difference $L_{21}(u_1,u_2) > L_{12}(u_1,u_2)$ across all quantile levels reveals an asymmetric concentration pattern. Specifically, losses accumulate faster among high-premium contracts than premiums accumulate among high-loss contracts. This finding suggests that the pricing structure may not fully account for the tail risk, as the most expensive policies generate disproportionately high losses relative to their premium share. The estimated 3D surfaces of $L_{12}(u_1,u_2)$ and $L_{21}(u_1,u_2)$ are shown in Figure \ref{3d_act_L12} and \ref{3d_act_L21}.  The slice plots are also shown in Figure \ref{slice_act_L12} and \ref{slice_act_L21} to clarify the behavior. For fixed values of $u_2$, the curves of $L_{12}(u_1,u_2)$ are increasing and convex, with larger values of $u_2$ producing uniformly higher curves. Similarly, the slice plots of $L_{21}(u_1,u_2)$ reveal increasing concentration as $u_1$ increases. The separation between the curves becomes more pronounced in the upper tail, highlighting the importance of extreme contracts in determining the overall concentration structure.
\begin{figure}[ht]
    \centering
    \includegraphics[scale=0.4]{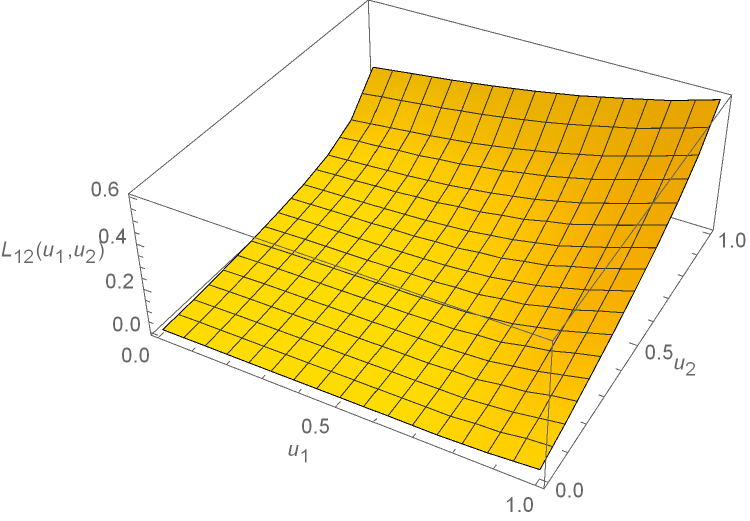}
    \caption{Estimated VBLS surface $\hat{L}_{12}(u_1,u_2)$ for workers compensation data}
    \label{3d_act_L12}
\end{figure}
\begin{figure}[ht]
    \centering
    \includegraphics[scale=0.4]{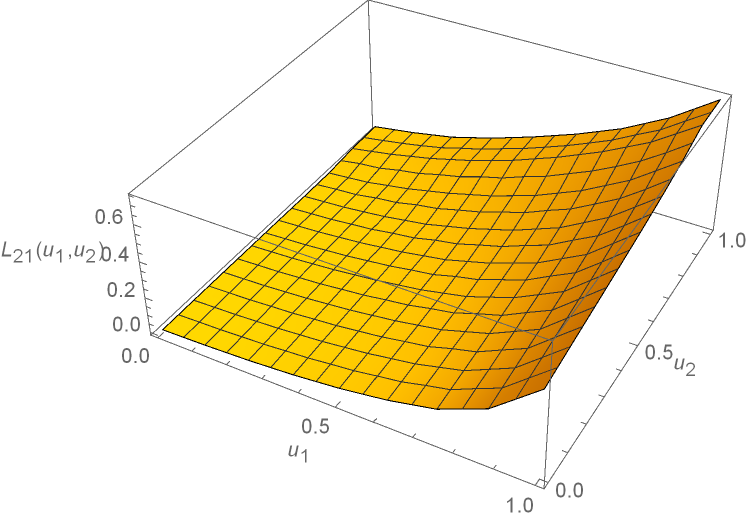}
    \caption{Estimated VBLS surface $\hat{L}_{21}(u_1,u_2)$ for workers compensation data }
    \label{3d_act_L21}
\end{figure}
\begin{figure}[ht]
    \centering
    \includegraphics[scale=0.4]{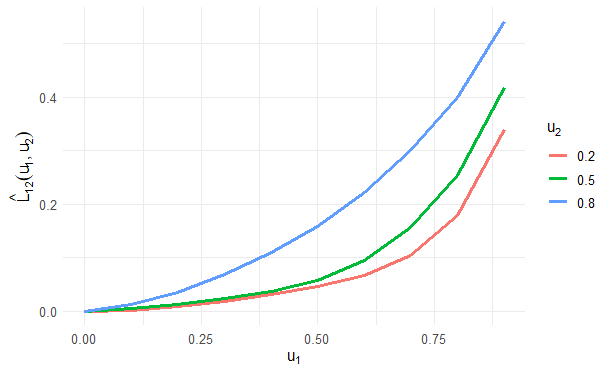}
    \caption{Slice plots of $\hat{L}_{12}(u_1,u_2)$ for workers compensation data }
    \label{slice_act_L12}
\end{figure}
\begin{figure}[ht]
    \centering
    \includegraphics[scale=0.4]{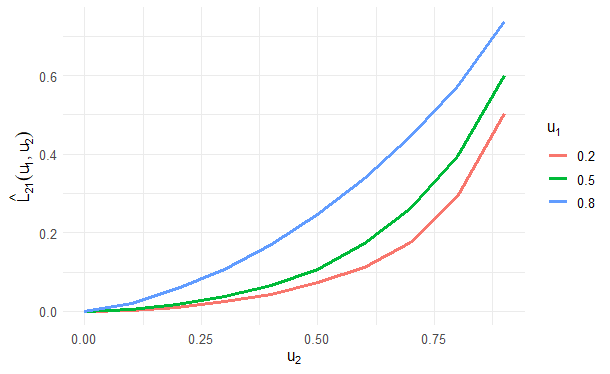}
    \caption{Slice plots of $\hat{L}_{21}(u_1,u_2)$ for workers compensation data }
    \label{slice_act_L21}
\end{figure}
The proposed bivariate Gini indices were estimated as $G_{12}=0.7750 \text{ and }
G_{21}=0.6737$. The large values of both indices indicate substantial concentration in the joint distribution of premiums and losses. From an actuarial perspective, the larger $G_{12}$ compared to $G_{21}$ indicates that premium inequality is more sensitive to tail events than loss inequality. Insurers may need to adjust risk loadings for high premium contracts to better reflect the observed concentration patterns. 

The actuarial application demonstrates that the proposed VBLS framework is capable of revealing directional concentration patterns that cannot be captured by marginal inequality measures alone. While the classical Gini coefficients indicate that both premiums and losses are highly concentrated, the VBLS surfaces and associated bivariate Gini indices provide additional information regarding how this concentration evolves under conditional upper tail regimes. Consequently, the proposed methodology offers a richer characterization of risk concentration within insurance portfolios and complements existing Lorenz curve based tools commonly employed in actuarial analysis.
\section*{Conflict of interest statement}
On behalf of all authors, the corresponding author states that there is no conflict of interest.

\bibliographystyle{apalike}
\bibliography{myref}
\end{document}